\documentclass[reqno,12pt]{amsart}

\usepackage{amsmath,amssymb,mathtools,verbatim,slashed,upgreek,bbm,appendix,wasysym}
\usepackage[protrusion=true,babel=true]{microtype}
\usepackage[english]{babel}
\usepackage[widespace]{fourier}
\usepackage[backrefs]{amsrefs}
\usepackage[margin=1.25in]{geometry}
\usepackage[onehalfspacing]{setspace}
\usepackage[pdfusetitle,pagebackref]{hyperref}
\numberwithin{equation}{section}                
\usepackage[nameinlink,noabbrev]{cleveref}
\expandafter\def\csname ver@etex.sty\endcsname{3000/12/31}

\usepackage{autonum}                         

\allowdisplaybreaks[1]

\let\originalleft\left
\let\originalright\right
\renewcommand{\left}{\mathopen{}\mathclose\bgroup\originalleft}
\renewcommand{\right}{\aftergroup\egroup\originalright}

\makeatletter
\renewcommand*{\eqref}[1]{\hyperref[{#1}]{\textup{\tagform@{\ref*{#1}}}}}
\makeatother

\newcommand*{\eqdef}{\mathrel{\vcenter{\baselineskip0.5ex \lineskiplimit0pt\hbox{.}\hbox{.}}}=}

\newtheorem*{acknowledgment}{Acknowledgment}

\newtheorem{theorem}{Theorem}[section]

\newtheorem{lemma}[theorem]{Lemma}
\newtheorem{corollary}[theorem]{Corollary}
\newtheorem{remark}[theorem]{Remark}
\newtheorem{definition}[theorem]{Definition}

\crefname{theorem}{Theorem}{Theorems}                 
\creflabelformat{theorem}{#2{#1}#3}                   
\crefname{main}{Main Theorem}{Main Theorems}          
\creflabelformat{main}{#2{#1}#3}                   
\crefname{lemma}{Lemma}{Lemmas}                       
\creflabelformat{lemma}{#2{#1}#3}                     
\crefname{corollary}{Corollary}{Corollaries}          
\creflabelformat{corollary}{#2{#1}#3}                 
\crefname{ineq}{inequality}{inequalities}             
\creflabelformat{ineq}{#2{\upshape(#1)}#3}               
\crefname{diag}{diagram}{diagrams}             
\creflabelformat{diag}{#2{\upshape(#1)}#3}               
\crefname{cond}{condition}{conditions}                
\creflabelformat{cond}{#2{#1}#3}                   
\crefname{table}{Table}{Tables}                       
\creflabelformat{table}{#2{\upshape(#1)}#3}              
\crefname{hypothesis}{Hypothesis}{Hypotheses}            
\creflabelformat{hypothesis}{#2{#1}#3}                
\crefname{remark}{Remark}{Remarks}                    
\creflabelformat{remark}{#2{#1}#3}                    
\crefname{definition}{Definition}{Definitions}           
\creflabelformat{def}{#2{#1}#3}                       

\def\id{\mathbbm{1}}
\def\cx{\mathbbm{C}}
\def\rl{\mathbbm{R}}
\def\N{\mathbbm{N}}

\def\Z{\mathbbm{Z}}

\def\bH{\mathbb{H}}

\def\cB{\mathcal{B}}
\def\cC{\mathcal{C}}

\def\cE{\mathcal{E}}

\def\cH{\mathcal{H}}

\def\cL{\mathcal{L}}
\def\cM{\mathcal{M}}

\def\cV{\mathcal{V}}

\def\Ad{\mathrm{Ad}}

\def\rd{\mathrm{d}}

\def\dist{\mathrm{dist}}

\def\dom{\mathrm{dom}}
\def\End{\mathrm{End}}

\def\Hom{\mathrm{Hom}}

\def\Jac{\mathrm{Jac}}

\def\PU{\mathrm{PU}}

\def\SU{\mathrm{SU}}
\def\supp{\mathrm{supp}}
\def\tr{\mathrm{tr}}
\def\rU{\mathrm{U}}

\def\vol{\: \mathrm{vol}}

\def\rd{\mathrm{d}}

\def\Hol{\mathrm{Hol}}

\def\cpt{\mathrm{cpt}}
\def\irr{\mathrm{irr}}

\def\rhoeq{\left[ \varrho \right]}
\def\nablaeq{\left[ \nabla \right]}

\def\whnabla{\widehat{\nabla}}
\def\whDelta{\widehat{\Delta}}
\def\whHol{\widehat{\Hol}}

\def\fB{\mathfrak{B}}
\def\fC{\mathfrak{C}}
\def\fU{\mathfrak{U}}

\title[Hyperbolic Bloch transforms and surface groups]{On the hyperbolic Bloch transform}
\date{\today}
\keywords{noncommutative Bloch transform, hyperbolic band theory, hyperbolic lattice, quantum condensed matter, hyperbolic geometry, surface groups}
\subjclass[2020]{35Q40, 42B37, 43A30, 43A85, 74E15, 81Q35}

\author{\'Akos Nagy}
\address[\'Akos Nagy]{Department of Mathematics, University of California, Santa Barbara}
\urladdr{\href{https://akosnagy.com/}{akosnagy.com}}
\email{\href{mailto:contact@akosnagy.com}{contact@akosnagy.com}}

\author{Steven Rayan}
\address[Steven Rayan]{Centre for Quantum Topology and Its Applications (quanTA) and Department of Mathematics \& Statistics, University of Saskatchewan}
\urladdr{\href{https://researchers.usask.ca/steven-rayan}{researchers.usask.ca/steven-rayan}}
\email{\href{mailto:rayan@math.usask.ca}{rayan@math.usask.ca}}

\hypersetup{
   unicode        = true,
   pdffitwindow   = true,
   pdftoolbar     = false,
   pdfmenubar     = false,
   pdfstartview   = {FitH},
   hypertexnames  = true,
   colorlinks     = true,
   linkcolor      = black,
   citecolor      = black,
   filecolor      = black,
   urlcolor       = blue
}

\calclayout
\begin{document}

\begin{abstract}
   Motivated by recent theoretical and experimental developments in the physics of hyperbolic crystals, we study the noncommutative Bloch transform of Fuchsian groups that we call the hyperbolic Bloch transform.

   First, we prove that the hyperbolic Bloch transform is injective and ``asymptotically unitary'' already in the simplest case, that is when the Hilbert space is the regular representation of the Fuchsian group, $\Gamma$. Second, when $\Gamma \subset \mathrm{PSU} (1, 1)$ acts isometrically on the hyperbolic plane, $\mathbb{H}$, and the Hilbert space is $L^2 \left( \mathbb{H} \right)$, then we define a modified, geometric Bloch transform, that sends wave functions to sections of stable, flat bundles over $\Sigma = \mathbb{H} / \Gamma$ and transforms the hyperbolic Laplacian into the covariant Laplacian.
\end{abstract}

\maketitle

\section{Introduction}

Bloch's theorem is a tool at the border of analysis, representation theory, and mathematical physics that is crucial for understanding the propagation of wave functions in a periodic medium under a potential that respects the symmetries of the system. The theorem applies when the medium under consideration is a tiling of Euclidean space (of any dimension) by copies of a regular, compact fundamental cell, which is an idealized model for any crystalline material. The theorem guarantees that a general wave function can be expressed in a basis of quasi-periodic waves, each of which is determined by periodic boundary conditions up to a choice of phase factor. These phase factors are given by $\rU (1)$-representations of the discrete, abelian translation group of the tiling --- equivalently, of the fundamental group of the torus $\mathbb{T}$ realized as the quotient of Euclidean space by translations. The space of these representations is the \emph{Brillouin zone} of the medium. In physical language, it is a space of ``crystal momenta''.

While the classical Bloch's theorem can be applied to any periodic physical problem, it has been the most influential in condensed matter physics where it describes the quantum-mechanical propagation of electrons in a periodic potential. It underlies the electronic band theory that forms the modern basis for the study of conductivity in solid-state physics. In the presence of an atypical dependence of fundamental quantum numbers, e.g. charge, spin, orbit, and lattice, thereby leading to so-called \emph{topological phases}, Bloch's theorem and electronic band theory facilitate a topological characterization of insulators, metals, and semimetals.

Very recently, interest has grown in $2$-dimensional electronic materials modeled on hyperbolic lattices rather than Euclidean ones, due in part to the artificial engineering of photonic hyperbolic materials \cite{KFH19}, the establishment of a mathematical hyperbolic band theory \cites{MR21,MR22}, and subsequent experiments involving negative curvature in an electronic context. These include the recent electric-circuit realization of a hyperbolic drum \cite{DRUM21}, whereby signals are observed to travel along what are apparently hyperbolic geodesics. Further theoretical and mathematical works emerging in hyperbolic lattices and band theory include \cites{ikeda_hyperbolic_2021,aoki_algebra_2021,boettcher_crystallography_2021,stegmaier_universality_2021,attar_selberg_2022,KR22,bienias_circuit_2022}.

As initiated in \cite{MR21}, one may consider a Fuchsian group $\Gamma \subset \mathrm{PSU} (1,1)$ acting on the hyperbolic plane $\bH$ such that $\Sigma \eqdef \bH / \Gamma$ is a compact, connected genus-$g$ algebraic curve for some $g \geqslant 2$. This gives rise to a complex momentum space parametrizing $\rU (1)$-representations of $\Gamma$. This space is the \emph{Jacobian} variety $\Jac \left( \Sigma \right) \cong \rU (1)^{2g}$. The Jacobian is a $g$-dimensional complex torus that is precisely the moduli space of holomorphic line bundles on $\Sigma$. This generalizes the Euclidean picture in which the complex momentum space is the dual elliptic curve, which is exactly the Jacobian in the $g = 1$ setting. In the present paper, we do not appeal to the specific complex structure on the Jacobian and use the term to refer to the underlying topological space $\rU (1)^{2g}$ without confusion. In other words, we need only regard the Jacobian as the topological space of $\rU (1)$-representations of $\Gamma \cong \pi_1 \left( \Sigma \right)$. The persistence of a compact space of crystal momenta begs the question of whether Bloch's theorem holds. For finite hyperbolic lattices with symmetric potentials, it is shown in \cite{MR22} that there is indeed a Bloch decomposition of general wave functions and that the decomposition involves higher rank irreducible representations of the fundamental group, which can be interpreted in algebraic geometry as moduli spaces of stable holomorphic bundles. In other words, there exist ``higher'' nonabelian Brillouin zones in the hyperbolic setting. Heuristically, by taking the ``large $N$ limit'' as the number of cells, $N$, goes to infinity, one obtains a Bloch's theorem in which all possible higher Brillouin zones appear.

We take a moment to recall the heuristics of Bloch transforms. Let $G$ be a group with Haar measure, $\mu_H$, and assume that the canonical action of $G$ on the Hilbert space $\ell^2 \left( \Gamma \right)$ is such that we have the direct integral decomposition
\begin{equation}
   \ell^2 \left( \Gamma \right) = \int\limits_{\widehat{G}}^\oplus \cH_{\rhoeq} \rd \widehat{\mu} \left( \rhoeq \right),
\end{equation}
where $\widehat{G}$ is the space of irreducible, finite dimensional representations of $G$, $\cH_{\rhoeq}$ is the representation space of $\varrho$, and $\widehat{\mu}$ is the canonical measure on $\widehat{G}$. Then, the Bloch transform of $\uppsi \in \ell^2 \left( \Gamma \right) \otimes \cH_0$ at $\rhoeq \in \widehat{G}$ (where $\cH_0$ is an auxiliary Hilbert space, e.g. the space of square integrable functions over a cell of a lattice) is, informally,
\begin{equation}
   \cB \left( \uppsi \right) \left( \rhoeq \right) = \int\limits_G \uppsi \left( \gamma \right) \varrho \left( \gamma \right) \rd \mu_H \left( \gamma \right) \in \End \left( \cH_{\rhoeq} \right).
\end{equation}
This construction has many advantages, in particular that $\Gamma$-periodic operators on $\cH$ transform into operators on $\cH_{\rhoeq}$. However, it is generally hard to say anything in the algebraic/analytic properties of $\cB \left( \uppsi \right)$. Moreover, the usual interpretation of the Bloch transform is that $\cB \left( \uppsi \right) \left( \rhoeq \right)$ is the \emph{$\rhoeq$-quasi-periodic} (or \emph{isotypical}) component of $\uppsi$, that is
\begin{equation}
   \cB \left( \gamma \cdot \uppsi \right) \left( \rhoeq \right) = \varrho \left( \gamma \right)^{- 1} \cB \left( \uppsi \right) \left( \rhoeq \right).
\end{equation}
This can then be used to study spectra of operators, such as Schr\"odinger operators on Euclidean spaces, as it allows us to decompose hard-to-understand partial differential operators on noncompact domains into easier-to-understand operators on compact domains with quasi-periodic boundary conditions. Unfortunately, in general, it is not \emph{a priori} clear whether $\cB$ is injective or how it interacts with the geometry of $\cH$. However, in sufficiently nice cases --- for example when $G$ is a finitely generated abelian group or is of type I --- $\cB$ will satisfy a Plancherel-type theorem, and so will be a unitary isomorphism of Hilbert spaces.

The goal of this paper is to rigorously study the analytic properties of the hyperbolic Bloch transform. More concretely, we prove that the hyperbolic Bloch transform is injective and that the Hilbert norm of $\uppsi$ can be recovered from it. We prove this by using a recent result of Magee in \cite{magee_random_II_2021}. In the case of the hyperbolic Bloch transform, we further refine this interpretation by modifying $\cB$ so that the resulting isotypical components can be viewed as sections of (stable, flat) Hermitian vector bundles over $\Sigma$. This construction uses Donaldson's theorem that states that $n$-dimensional irreducible representations of $\Gamma$ correspond bijectively to stable, flat, metric compatible connections on the Hermitian vector bundle $E = \Sigma \times \cx^n$.

\smallskip

\subsection*{Organization of the paper:} In \Cref{sec:group}, we introduce relevant ideas about the symmetry group $\Gamma$ and the hyperbolic analogue of the Brillouin zone, $\cM_\Gamma$. In \Cref{sec:abstract}, we define the abstract Bloch transform corresponding to $\Gamma$ and in \Cref{theorem:asymptotic_almost_unitarity} we prove that is injective and ``almost unitary''. In \Cref{sec:HBT}, we define the hyperbolic Bloch transform on $L^2 \left( \bH \right)$ and study its properties. Finally, in \Cref{sec:twisted}, we study the case when a periodic gauge field is present on the hyperbolic plane.

\smallskip

\begin{acknowledgment}
   \emph{The authors are grateful to Elliot Kienzle, Semyon Klevtsov, Michael Magee, Rafe Mazzeo, Joseph Maciejko, Gon\c{c}alo Oliveira, and Jacek Szmigielski for fruitful discussions at various stages of the work. The second-named author was partially funded by a Natural Sciences and Engineering Research Council of Canada (NSERC) Discovery Grant, a Canadian Tri-Agency New Frontiers in Research Fund (NFRF) Exploration Stream Grant, and a Pacific Institute for the Mathematical Sciences (PIMS) Collaborative Research Group Award. He also acknowledges the Institut de Recherche Math\'ematique Avanc\'ee (Strasbourg) and the organizers of the ``Quantum Hall Effect and Topological Phases'' conference (June 2022) for facilitating useful discussions with the participants. Both authors acknowledge the hospitality of the American Institute of Mathematics (San Jose), where final steps in this work were taken in July 2022 during the ``Geometry and Physics of ALX Metrics'' workshop that they co-organized together with Laura Fredrickson and Hartmut Wei\ss.}
\end{acknowledgment}

\bigskip

\section{Symmetry groups and measures}
\label{sec:group}

Fix $g \in \N \cap \left[ 2, \infty \right)$, the \emph{genus}, and let $\Gamma$ be the discrete group generated by the elements $\alpha_1, \alpha_2, \ldots, \alpha_g$ and $\beta_1, \beta_2, \ldots, \beta_g$, satisfying the multiplicative relation
\begin{equation}
   \left[ \alpha_1, \beta_1 \right] \left[ \alpha_2, \beta_2 \right] \cdots \left[ \alpha_g, \beta_g \right] = \id. \label{eq:relation}
\end{equation}
Let
\begin{equation}
   \cM_\Gamma \eqdef \bigcup\limits_{n \in \N_+}^\infty \overbrace{\Hom_{\irr} \left( \Gamma, \rU (n) \right) / \rU (n)}^{\cM_\Gamma^n \eqdef}.
\end{equation}
For any $\rhoeq \in \cM_\Gamma^n$, we define the dimension as $\dim \left( \rhoeq \right) = n$. Note that the dimension, regarded as a function on $\cM_\Gamma$, is not assumed to be bounded above.

Since $\rU (1)$ is abelian and every 1-dimensional representation is irreducible, we have that
\begin{equation}
   \cM_\Gamma^1 = \Hom_{\irr} \left( \Gamma, \rU (1) \right) / \rU (1) = \Hom \left( \Gamma, \rU (1) \right),
\end{equation}
which is isomorphic to the abelian group, $\Jac_g \eqdef \rU (1)^{2g}$, the \emph{Jacobian}, via the map
\begin{equation}
   \Hom \left( \Gamma, \rU (1) \right) \rightarrow \Jac_g; \: \lambda \mapsto \left( \lambda \left( \alpha_1 \right), \lambda \left( \beta_1 \right), \lambda \left( \alpha_2 \right), \lambda \left( \beta_2 \right), \ldots, \lambda \left( \alpha_g \right), \lambda \left( \beta_g \right) \right).
\end{equation}
This map is surjective because of \cref{eq:relation}.

\smallskip

We allow ourselves a quick digression about special unitary representations. By \cite{narashimhan_stable_1965}, the space $\Hom_{\irr} \left( \Gamma, \SU (n) \right)$ is diffeomorphic to the moduli space of stable holomorphic vector bundles of rank $n$ and degree zero over $\Sigma$. Using this fact, together with Yang--Mills theory, Atiyah and Bott showed in \cite{atiyah_yang_1983} that the (dense) smooth locus of $\Hom_{\irr} \left( \Gamma, \SU (n) \right)$ is equipped with the structure of a (finite volume) K\"ahler manifold (of complex dimension $n^2 \left( \mathrm{genus} \left( \Sigma \right) - 1 \right) + 1$), and this K\"ahler structure of canonical, up to a global factor. In \cite{goldman_symplectic_1984}, Goldman also introduces a symplectic form directly on the smooth locus of $\Hom_{\irr} \left( \Gamma, \SU (n) \right)$, which coincides up to a global factor with the Atiyah--Bott K\"ahler form. Using these structures, one can extend the volume measure to all of $\Hom_{\irr} \left( \Gamma, \SU (n) \right)$, thus achieving a finite-volume measure space. Note that the canonical projection $\Hom_{\irr} \left( \Gamma, \SU (n) \right) \rightarrow \Hom_{\irr} \left( \Gamma, \PU (n) \right)$, is a $\Z_n^{2g}$-cover and the measure of Atiyah--Bott and Goldman is invariant under deck transformations. Thus this measure descends to $\Hom_{\irr} \left( \Gamma, \PU (n) \right)$. We call this measure $\widetilde{\mu}_n$ and assume, without any loss of generality, that $\widetilde{\mu}_n \left( \Hom_{\irr} \left( \Gamma, \PU (n) \right) \right) = 1$.

Returning to the geometry of $\cM_\Gamma^n$ and employing the Haar measure on $\rU (n)$, we equip $\cM_\Gamma^n$ with a measure, $\mu_n$, which is again unique up to a global factor. Furthermore, $\mu_n \left( \cM_\Gamma^n \right) < \infty$, so without any loss of generality we can assume that $\mu_n \left( \cM_\Gamma^n \right) = \tfrac{1}{n}$. This normalization is convenient for the computations of this paper. When $n = 1$, this measure is just the (normalized) Lebesgue measure on $\cM_1 = \Jac_g = \rU (1)^{2g}$. When $n \geqslant 2$, then $\mu_n$ can be interpreted in a more geometric way. The canonical projection
\begin{equation}
   \pi_n : \cM_\Gamma^n \rightarrow \Hom_{\irr} \left( \Gamma, \PU (n) \right), \label{eq:M_n-fibration}
\end{equation}
makes $\cM_\Gamma^n$ a principal $\Jac_g$-bundle. Using $\mu_1$ and the measure of Atiyah--Bott and Goldman, one can then construct a ``compatible'' measure on $\cM_\Gamma^n$, that is if $A \subseteq \cM_\Gamma^n \Jac_g$, $B \subseteq \cM_1$, and $C \subseteq \Hom_{\irr} \left( \Gamma, \PU (n) \right)$ are measurable sets, such that $A \cong B \times C$, then $\mu_n \left( A \right) = \tfrac{1}{n} \mu_1 \left( B \right) \widetilde{\mu}_n \left( C \right)$.

\smallskip

Next, we introduce an important identity about $\mu_n$. For each $\rhoeq \in \cM_\Gamma$, let
\begin{equation}
   \chi_{\rhoeq} \eqdef \tr \left( \varrho \left( \cdot \right) \right) : \: \Gamma \rightarrow \cx,
\end{equation}
be the \emph{character} of $\rhoeq$ (which depends only on the class of $\rho$). For each $\gamma \in \Gamma$, let
\begin{equation}
   \widehat{\gamma} : \cM_\Gamma \rightarrow \cx; \: \rhoeq \mapsto \widehat{\gamma}\left( \rhoeq \right) \eqdef \chi_{\rhoeq} \left( \gamma \right).
\end{equation}
Let $\widehat{\gamma}_n \eqdef \widehat{\gamma}|_{\cM_\Gamma^n}$. Let us now define the quantity
\begin{equation}
   \mathbb{E}_n \left( \gamma \right) \eqdef \int\limits_{\cM_\Gamma^n} \widehat{\gamma}_n \rd \mu_n, \label{eq:E_gn_def}
\end{equation}
Using Fubini's theorem and the definition of $\mu_n$, one can evaluate $\mathbb{E}_n \left( \gamma \right)$ by first integrating on the fibers of \eqref{eq:M_n-fibration}. Then, if $\gamma$ is not in the commutator group of $\Gamma$, then $\mathbb{E}_n \left( \gamma \right) = 0$. When $\gamma \in \left[ \Gamma, \Gamma \right]$, then the integrand in \cref{eq:E_gn_def} is invariant under the action of $\cM_1$, thus if we define
\begin{equation}
   e_{n, \gamma} : \Hom_{\irr} \left( \Gamma, \PU (n) \right) \mapsto \cx ; \: \pi_n \left( \rhoeq \right) \mapsto \left\{ \begin{array}{rr} \widehat{\gamma}_n \left( \rhoeq \right), & \gamma \in \left[ \Gamma, \Gamma \right], \\ 0, & \gamma \notin \left[ \Gamma, \Gamma \right], \end{array} \right.
\end{equation}
then $e_{n, \gamma}$ is a $\widetilde{\mu}_n$-integrable function that satisfies
\begin{equation}
   \mathbb{E}_n \left( \gamma \right) = \int\limits_{\Hom_{\irr} \left( \Gamma, \PU (n) \right)} e_{n, \gamma} \rd \widetilde{\mu}_n.
\end{equation}
Now pulling back to an integral on $\Hom_{\irr} \left( \Gamma, \SU (n) \right)$ and using \cite{magee_random_I_2022}*{Theorem~1.1} and \cite{magee_random_II_2021}*{Theorem~1.2}, we get the following: For each $\gamma \in \Gamma - \{ e \}$ and $k \in \N_+$, there are (rational) numbers $a_k \left( \gamma \right)$, such that for all $M \in \N_+$, there is a positive number $C_{\gamma, M}$, so that that
\begin{equation}
   \left| \mathbb{E}_n \left( \gamma \right) - \sum\limits_{k = 1}^M \frac{a_k \left( \gamma \right)}{n^k} \right| \leqslant \frac{C_{\gamma,M}}{n^{M + 1}}. \label{eq:chi_expected_value}
\end{equation}

\bigskip

\section{The abstract Bloch transform for surface groups}
\label{sec:abstract}

For each $n \in \N_+$, we define a Hermitian vector bundle, $\cV^n$, over $\cM_\Gamma^n$ as follows: for all $\varrho \in \Hom_{\irr} \left( \Gamma, \rU (n) \right)$, $A \in \End \left( \cx^n \right)$, and $U \in \rU (n)$, let
\begin{equation}
   U \left( \varrho, A \right) \eqdef \left( \Ad \left( U \right) \circ \varrho, \left( U A U^* \right) \right),
\end{equation}
which defines an a free, linear action of $\PU (n)$ on $\Hom_{\irr} \left( \Gamma, \rU (n) \right) \times \End \left( \cx^n \right)$, covering the free action of $\PU (n)$ on $\Hom_{\irr} \left( \Gamma, \rU (n) \right)$. Thus
\begin{equation}
   \cV^n \eqdef \left( \Hom_{\irr} \left( \Gamma, \rU (n) \right) \times \End \left( \cx^n \right) \right) / \PU (n) \rightarrow \cM_\Gamma^n, \label{eq:E_def}
\end{equation}
defines a vector bundle over $\cM_\Gamma^n$ with fibers isomorphic to $\End \left( \cx^n \right)$. Using the Hilbert--Schmidt norm on $\End \left( \cx^n \right)$, we equip $\cV^n$ with a Hermitian metric. If $\left[ \varrho, A \right] \in \cV_{\rhoeq}^n$ and $\gamma \in \Gamma$, then let
\begin{equation}
   \tau_\gamma \left( \left[ \varrho, A \right] \right) \eqdef \tr \left( \varrho \left( \gamma \right)^* A \right).
\end{equation}
Now $\tau_\gamma$ is a well defined section of $\left( \cV^n \right)^*$. Let $L^2 \left( \cV^n \right)$ be the space of sections, $\widehat{\uppsi}$, of $\cV^n$, such that $|\widehat{\uppsi}| \in L^2 \left( \cM_\Gamma^n, \mu_n \right)$. Similarly, if $U \subseteq \cM_\Gamma^n$ is a $\mu_n$-measurable set, then let $L^2 \left( \cV^n \right)$ be the space of sections, $\widehat{\uppsi}$, of $\cV^n\big|_U$, such that $\int_U |\widehat{\uppsi}|^2 \mu_n < \infty$. The assignment $U \mapsto L^2 \left( \cV^n\big|_U \right)$ then defines a Hilbert sheaf over $\cM_\Gamma$ which we call $\cV$.

Finally, let
\begin{equation}
   C_{\cpt} \left( \Gamma \right) \eqdef \left\{ \: \uppsi : \Gamma \rightarrow \cx \: \middle| \: |\Gamma - \uppsi^{- 1} (0)| < \infty \: \right\},
\end{equation}
be the vector space of compactly-supported, complex-valued functions on $\Gamma$. For all $p \in [1, \infty) \cup \{ \infty \}$, let $\ell^p \left( \Gamma \right)$ be defined in the obvious way. Of course, $C_{\cpt} \left( \Gamma \right) \subset \ell^1 \left( \Gamma \right) \subset \ell^p \left( \Gamma \right)$.

\begin{definition}[Abstract Bloch transform]
   Let the \emph{abstract Bloch transform} be the map $\cB$ such that for all $\uppsi \in C_{\cpt} \left( \Gamma \right)$ and $\rhoeq \in \cM_\Gamma$, we have
   \begin{equation}
      \cB \left( \uppsi \right) \left( \rhoeq \right) \eqdef \left[ \varrho, \sum\limits_{\gamma \in \Gamma} \uppsi \left( \gamma \right) \varrho \left( \gamma \right) \right] \in \cV_{\rhoeq}^{\dim \left( \rhoeq \right)}. \label{eq:abstract_Bloch_def}
   \end{equation}
\end{definition}

\smallskip

\begin{remark}
   Before we proceed, let us outline what is known about Bloch transforms of discrete but not necessarily finite groups, including surface groups in particular. The definitions in \cref{eq:E_def,eq:abstract_Bloch_def} do not require $\Gamma$ to be a surface group. In fact, similar abstract Bloch transforms have been studied for a long time. An excellent summary of the general theory can be found in \cite{gruber_noncommutative_2001}. Bloch transforms of the form \eqref{eq:abstract_Bloch_def} can be used to study properties of $\Gamma$-symmetric operators, as it transforms hard-to-understand (e.g. differential) operators into bundle maps (algebraic operators) on $\cV$; cf \cite{sunada_fundamental_1994} for the general idea and \cite{marcolli_twisted_1999} for surface groups. However, the full strength of the Bloch transform in the Euclidean case comes from the fact that $\cB$ is a unitary isomorphism of Hilbert spaces. This statement can be generalized to (potentially noncommutative) groups whose group algebra is a Neumann algebra of Type $I$, as the harmonic analysis is well understood in that setting; see \cite{kocabova_generalized_2008} for the Bloch theory of these groups. Unfortunately, the group algebra of surface groups (of genus at least two) is a full factor of Type $\textit{II}_1$; cf. \cite{akemann_operator_1981}.

   Finally, note that if $C_{\cpt} \left( \Gamma \right)$ is replaced by the space $C_{\cpt} \left( \Gamma \right) \otimes \cH_0$, where $\cH_0$ is a Hilbert space (that can be thought of as the space of modes in a cell), then we can define
   \begin{equation}
      \cB \left( \uppsi \right) \left( \rhoeq \right) \eqdef \left[ \varrho, \sum\limits_{\gamma \in \Gamma} \uppsi \left( \gamma \right) \varrho \left( \gamma \right) \right] \in \cV_{\rhoeq}^{\dim \left( \rhoeq \right)} \otimes \cH_0. \label{eq:abt_with_modes}
   \end{equation}
   Note that if $\uppsi \in \ell^1 \left( \Gamma \right)$, then the right hand side of \cref{eq:abt_with_modes} is absolutely convergent; however, it is not \emph{a priori} clear if it is well-defined for (almost) all $\rhoeq \in \cM_\Gamma$, when $\uppsi \in \ell^2 \left( \Gamma \right) - \ell^1 \left( \Gamma \right)$.

   Thus most of the technical difficulties are present in the base case of \cref{eq:abstract_Bloch_def}, that is when $\cH_0 = \cx$.
\end{remark}

\smallskip

Next, let us make an immediate observation about the abstract Bloch transform:

\begin{lemma}
   We can extend $\cB_n$ to a continuous linear map
   \begin{equation}
      \cB_n : \ell^1 \left( \Gamma \right) \rightarrow L^2 \left( \cV^n \right); \: \uppsi \mapsto \cB \left( \uppsi \right)\big|_{\cM_\Gamma^n} \left( \cdot \right),
   \end{equation}
   between Banach spaces. In fact, $\| \cB_n \| = 1$.
\end{lemma}

\begin{proof}
   If $\uppsi \in C_{\cpt} \left( \Gamma \right)$ and $n \in \N_+$, then
   \begin{align}
      \int\limits_{\cM_\Gamma^n} \left| \cB \left( \uppsi \right) \left( \rhoeq \right) \right|^2 \rd \mu_n \left( \rhoeq \right)   &= \sum\limits_{\gamma_1, \gamma_2 \in \Gamma} \overline{\uppsi \left( \gamma_1 \right)} \uppsi \left( \gamma_2 \right) \int\limits_{\cM_\Gamma^n} \chi_{\rhoeq} \left( \gamma_1^{- 1} \gamma_2 \right) \rd \mu_n \left( \rhoeq \right) \\
         &\leqslant n \mu_n \left( \cM_\Gamma^n \right) \| \uppsi \|_{\ell^1 \left( \Gamma \right)}^2 \\
         &= \| \uppsi \|_{\ell^1 \left( \Gamma \right)}^2,
   \end{align}
   thus as a densely defined operator from $\ell^1 \left( \Gamma \right)$ to $L^2 \left( \cV^n \right)$, we see that $\cB_n$ has operator norm at most one. In particular, it is continuous and can be extended continuously to all of $\ell^1 \left( \Gamma \right)$. Using $\uppsi$ that is supported on a single element of $\Gamma$, we see that the operator norm is exactly 1.
\end{proof}

\smallskip

\begin{corollary}
   Let $\iota : \ell^1 \left( \Gamma \right) \rightarrow \ell^\infty \left( \Gamma \right) = \left( \ell^1 \left( \Gamma \right) \right)^*$ be the natural continuous embedding. The map $\cB_n^* \cB_n - \iota : \ell^1 \left( \Gamma \right) \rightarrow \ell^\infty \left( \Gamma \right)$ is continuous.
\end{corollary}

\smallskip

\begin{remark}
   When regarded as a densely defined operator from $\ell^2 \left( \Gamma \right)$ to $L^2 \left( \cV^n \right)$, the adjoint of $\cB_n$ is given by
   \begin{equation}
      \cB_n^* \left( \widehat{\uppsi} \right) \left( \gamma \right) = \int\limits_{\cM_\Gamma^n} \tau_\gamma \left( \widehat{\uppsi} \right) \rd \mu_n, \label{eq:B_n_adjoint}
   \end{equation}
   where $\gamma \in \Gamma$ and $\dom \left( \cB_n^* \right) \subseteq L^2 \left( \cV^n \right)$ can be characterized as
   \begin{equation}
      \dom \left( \cB_n^* \right) = \left\{ \: \widehat{\uppsi} \in L^2 \left( \cV^n \right) \: \middle| \: \exists c \in \rl_+ : \forall \uppsi \in C_{\cpt} \left( \Gamma \right) : \: \left| \left\langle \widehat{\uppsi} \middle| \cB_n \left( \uppsi \right) \right\rangle_{L^2 \left( \cV^n \right)} \right| \leqslant c \| \uppsi \|_{\ell^2 \left( \Gamma \right)} \: \right\}. \label{eq:dom_cB_n_star}
   \end{equation}

   The operator $\cB_n$ has a continuous extension to $\ell^2 \left( \Gamma \right)$ exactly if $\dom \left( \cB_n^* \right) $ is dense in $L^2 \left( \cV^n \right)$, in which case the continuous extension is given by $\cB_n^{**}$. While we cannot prove this claim yet---but prove a weaker version of it in \Cref{theorem:asymptotic_almost_unitarity}---we conjecture that $\cB_n$ is continuous and \emph{asymptotically unitary}, that is $\cB_n^* \cB_n$ converges (in the operator norm topology) to the identity. In fact, \cref{eq:chi_expected_value} suggests that the rate of convergence should be $O \left( \tfrac{1}{n} \right)$. Such a generalization could be achieved by a stronger version of \cite{magee_random_II_2021}*{Theorem~1.2}.
\end{remark}

\medskip

In our first main result below, we study functional analytic and algebraic properties of $\cB_n$.

\begin{theorem}
   \label{theorem:asymptotic_almost_unitarity}
   For each $\uppsi \in \ell^1 \left( \Gamma \right)$, the sequence, $\left( \cB_n^* \cB_n \left( \uppsi \right) \right)_{n \in \N_+}$, converges to $\uppsi$ in the topology of $\ell^\infty \left( \Gamma \right)$.

   Furthermore, if $T_\gamma$ is the translation map on $\ell^i \left( \Gamma \right), i \in \left[ 1, \infty \right) \cup \{ \infty \}$, given by $\gamma \in \Gamma$, that is
   \begin{equation}
      \forall \gamma^\prime \in \Gamma : \: T_\gamma \left( \uppsi \right) \left( \gamma^\prime \right) = \uppsi \left( \gamma \gamma^\prime \right),
   \end{equation}
   and $\widehat{T}_\gamma$ is defined on section of $\cE$ via
   \begin{equation}
      \widehat{T}_\gamma \left( \left[ \varrho, A \right] \right) \eqdef \left[ \varrho, \varrho \left( \gamma \right)^{- 1} A \right],
   \end{equation}
   then both $T_\gamma$ and $\widehat{T}_\gamma$ (restricted to $\cE^n$) are unitary and we have that
   \begin{equation}
      \cB \circ T_\gamma = \widehat{T}_\gamma \circ \cB.
   \end{equation}
\end{theorem}

\begin{proof}
   Note that
   \begin{equation}
      \left| \mathbb{E}_n \left( \gamma \right) - \delta_{\gamma, e} \right| \leqslant 1,
   \end{equation}
   and, by \cref{eq:chi_expected_value}, there exists $c : \Gamma \rightarrow \rl_+$ such that for all $\gamma \in \Gamma$ and $n \in \N_+$
   \begin{equation}
      \left| \mathbb{E}_n \left( \gamma \right) - \delta_{\gamma, e} \right| \leqslant \frac{c (\gamma)}{n}.
   \end{equation}
   
   Let $\uppsi, \uppsi^\prime \in C_{\cpt} \left( \Gamma \right)$. Now we get that
   \begin{align}
      \left| \left\langle \cB_n \left( \uppsi \right) \middle| \cB_n \left( \uppsi^\prime \right) \right\rangle_{L^2 \left( \cV^n \right)} - \left\langle \uppsi \middle| \uppsi^\prime \right\rangle_{\ell^2 \left( \Gamma \right)} \right| &= \left| \sum\limits_{\gamma_1, \gamma_2 \in \Gamma} \overline{\uppsi \left( \gamma_1 \right)} \uppsi^\prime \left( \gamma_2 \right) \left( \mathbb{E}_n \left( \gamma_1^{- 1} \gamma_2 \right) - \delta_{\gamma_1^{- 1} \gamma_2, e} \right) \right| \\
            &\leqslant \sum\limits_{\gamma_1, \gamma_2 \in \Gamma} |\uppsi \left( \gamma_1 \right)| |\uppsi^\prime \left( \gamma_2 \right)| \\
            &= \| \uppsi \|_{l^1 \left( \Gamma \right)} \| \uppsi^\prime \|_{l^1 \left( \Gamma \right)}.
   \end{align}
   Thus for any fixed $\uppsi \in C_{\cpt} \left( \Gamma \right)$, there exists a continuous sequence of linear functionals, denoted by $\left( f_n^\uppsi \in \left( \ell^1 \left( \Gamma \right) \right)^* \right)_{n \in \N_+}$ and defined via
   \begin{equation}
      f_n^\uppsi \left( \uppsi^\prime \right) \eqdef \left\langle \cB_n \left( \uppsi \right) \middle| \cB_n \left( \uppsi^\prime \right) \right\rangle_{L^2 \left( \cV^n \right)} - \left\langle \uppsi \middle| \uppsi^\prime \right\rangle_{\ell^2 \left( \Gamma \right)},
   \end{equation}
   that is bounded (by one) in the operator norm.

   Let $\uppsi, \uppsi^\prime \in C_{\cpt} \left( \Gamma \right)$ and
   \begin{equation}
      K \left( \uppsi, \uppsi^\prime \right) \eqdef \sup \left( \left\{ \: \left| c \left( \gamma_1^{- 1} \gamma_2 \right) \right| \: \middle| \: \uppsi \left( \gamma_1 \right) \neq 0, \: \& \: \uppsi^\prime \left( \gamma_2 \right) \neq 0 \: \right\} \right), \label{eq:K_def}
   \end{equation}
   which is finite, since $\uppsi, \uppsi^\prime \in C_{\cpt} \left( \Gamma \right)$. Using \cref{eq:chi_expected_value,eq:K_def}, and that $\uppsi, \uppsi^\prime$ are bounded functions, we get
   \begin{align}
      \left| f_n^\uppsi \left( \uppsi^\prime \right) \right| &= \left| \sum\limits_{\gamma_1, \gamma_2 \in \Gamma} \overline{\uppsi \left( \gamma_1 \right)} \uppsi^\prime \left( \gamma_2 \right) \left( \mathbb{E}_n \left( \gamma_1^{- 1} \gamma_2 \right) - \delta_{\gamma_1^{- 1} \gamma_2, e} \right) \right| \\
            &\leqslant \frac{K \left( \uppsi, \uppsi^\prime \right)}{n} \sum\limits_{\gamma_1, \gamma_2 \in \Gamma} |\uppsi \left( \gamma_1 \right)| |\uppsi^\prime \left( \gamma_2 \right)| \\
            &= \frac{K \left( \uppsi, \uppsi^\prime \right)}{n} \| \uppsi \|_{l^1 \left( \Gamma \right)} \| \uppsi^\prime \|_{l^1 \left( \Gamma \right)}.
   \end{align}
   Thus for any fixed $\uppsi \in C_{\cpt} \left( \Gamma \right)$, the sequence of functionals, $\left( f_n^\uppsi \right)_{n \in \N_+}$ converges pointwise on a dense subset $C_{\cpt} \left( \Gamma \right) \subset \ell^1 \left( \Gamma \right)$. By the Banach--Steinhaus Theorem \cite{buhler_functional_2018}*{Theorem~2.1.5~(ii)}, $\left( f_n^\uppsi \right)_{n \in \N_+}$ converges pointwise on $\ell^1 \left( \Gamma \right)$ to a functional $f^\uppsi \in \left( \ell^1 \left( \Gamma \right) \right)^*$. Furthermore, using the above computation again, for all $\uppsi^\prime \in C_{\cpt} \left( \Gamma \right)$ we have that $f^\uppsi \left( \uppsi^\prime \right) = 0$, thus $f^\uppsi = 0$, by density. Hence the sequence, $\left( \cB_n^* \cB_n \left( \uppsi \right) - \uppsi \right)_{n \in \N_+}$, converges weakly to zero in $\left( \ell^1 \left( \Gamma \right) \right)^* = \ell^\infty \left( \Gamma \right)$, which proves the first claim.

   The claims about $T_\gamma$ and $\widehat{T}_\gamma$ follow from straightforward computations.
\end{proof}

\smallskip

\begin{corollary}
   \label{corollary:compact}
   The abstract Bloch transform, $\cB$, is injective and if $\uppsi, \uppsi^\prime \in C_{\cpt} \left( \Gamma \right)$, then
   \begin{equation}
      \lim\limits_{n \rightarrow \infty} \left\langle \cB_n \left( \uppsi \right) \middle| \cB_n \left( \uppsi^\prime \right) \right\rangle_{L^2 \left( \cV^n \right)} = \left\langle \uppsi \middle| \uppsi^\prime \right\rangle_{\ell^2 \left( \Gamma \right)}.
   \end{equation}
   In particular, the norm of $\uppsi \in C_{\cpt} \left( \Gamma \right)$ can be recovered from $\cB \left( \uppsi \right)$.
\end{corollary}

\smallskip

\begin{remark}
   In \cite{magee_random_II_2021}, Magee conjectures that for all $\gamma \in \left[ \Gamma, \Gamma \right]$, there is a positive number, $C$, such that if $|\gamma|$ is the commutator length of $\gamma$, then
   \begin{equation}
      \left| \mathbb{E}_n \left( \gamma \right) \right| \leqslant \frac{C}{n^{2 |\gamma|}}.
   \end{equation}
   Such a bound would be enough to prove that $\cB_n$ has a continuous extension to $\ell^2 \left( \Gamma \right)$ and that the sequence of operators, $\left( \cB_n^* \cB_n \right)_{n \in \N_+}$, converges pointwise to the identity of $\ell^2 \left( \Gamma \right)$. It is still not \emph{a priori} enough to show that the convergence holds in the operator norm.
\end{remark}

\smallskip

\begin{remark}
   Let $H$ be a linear operator on $C_\cpt \left( \Gamma \right)$ that is $\Gamma$-periodic
   \begin{equation}
      \forall \gamma \in \Gamma : \quad H \circ T_\gamma = T_\gamma \circ H,
   \end{equation}
   where $T_\gamma$ is the translation operator with respect to $\gamma$. Then we have that
   \begin{equation}
      \widehat{H}_n \eqdef \cB_n \circ H \circ \cB_n^*,
   \end{equation}
   is a bundle map from $\cB_n \left( C_\cpt \left( \Gamma \right) \right)$ that covers the identity of $\cM_\Gamma^n$, that is, if $\widehat{\uppsi}$ is a section in $\cB_n \left( C_\cpt \left( \Gamma \right) \right)$, $f$ is a (measurable) function on $\cM_\Gamma^n$, and $\rhoeq \in \cM_\Gamma^n$, then
   \begin{equation}
      \widehat{H}_n \left( f \widehat{\uppsi} \right) \left( \rhoeq \right) = f \left( \rhoeq \right) \widehat{H}_n \left( \widehat{\uppsi} \right) \left( \rhoeq \right).
   \end{equation}
   In particular, $\widehat{H}_n \left( \widehat{\uppsi} \right) \left( \rhoeq \right)$ is in $\cV_{\rhoeq}^n$ and it can be computed from $\widehat{\uppsi} \left( \rhoeq \right)$ alone.

   More generally, if there is an automorphism $\alpha : \Gamma \rightarrow \Gamma$, then it defines an automorphism of $\cM_\Gamma^n$, given by $\widehat{\alpha} \eqdef \left( \rhoeq \mapsto \rhoeq \circ \alpha^{- 1} \right)$. Now if $H$ is a continuous operator on $\ell^2 \left( \Gamma \right)$ such that
   \begin{equation}
      \forall \gamma \in \Gamma : \quad H \circ T_\gamma = T_{\alpha \left( \gamma \right)} \circ H,
   \end{equation}
   then we have that
   \begin{equation}
      \widehat{H}_n \eqdef \cB_n \circ H \circ \cB_n^*,
   \end{equation}
   is a bundle map from $\cB_n \left( C_\cpt \left( \Gamma \right) \right)$ that covers the map $\widehat{\alpha}$. More concretely if $\widehat{\uppsi}$ is a section of $\cV^n$, $f$ is a function on $\cM_\Gamma^n$, and $\rhoeq \in \cM_\Gamma^n$, then
   \begin{equation}
      \widehat{H}_n \left( f \widehat{\uppsi} \right) \left( \rhoeq \right) = f \left( \widehat{\alpha} \left( \rhoeq \right) \right) \widehat{H}_n \left( \widehat{\uppsi} \right) \left( \rhoeq \right).
   \end{equation}
   In particular, $\widehat{H}_n \left( \widehat{\uppsi} \right) \left( \rhoeq \right)$ is in $\cV_{\widehat{\alpha} \left( \rhoeq \right)}^n$ and it can be computed from $\widehat{\uppsi} \left( \rhoeq \right)$ alone.

   Finally, \Cref{theorem:asymptotic_almost_unitarity} implies that the sequence of operators, $\left( \cB_n^* \circ \widehat{H}_n \circ \cB_n \right)_{n \in \N_+}$, converges pointwise to $H$.
\end{remark}

\smallskip

\begin{remark}
   An immediate consequence of the general form of the Bloch transform \eqref{eq:abstract_Bloch_def} is a convolution theorem. Assume that $\cH_0$ is an algebra and thus so is $\cV_{\rhoeq}^{\dim \left( \rhoeq \right)} \otimes \cH_0$. Then we can define the convolution of $\uppsi_1, \uppsi_1 \in \cH$ as
   \begin{equation}
      \left( \uppsi_1 \star \uppsi_2 \right) \left( \gamma \right) \eqdef \sum\limits_{\gamma^\prime \in \Gamma} \uppsi_1 \left( \gamma^\prime \right) \cdot \uppsi_2 \left( \left( \gamma^\prime \right)^{- 1} \gamma \right),
   \end{equation}
   and straightforward computation shows that
   \begin{equation}
      \cB \left( \uppsi_1 \star \uppsi_2 \right) \left( \rhoeq \right) = \cB \left( \uppsi_1 \right) \left( \rhoeq \right) \cdot \cB \left( \uppsi_2 \right) \left( \rhoeq \right).
   \end{equation}
\end{remark}

\smallskip

Finally, by \Cref{corollary:compact}, almost all information about $\uppsi \in C_{\cpt} \left( \Gamma \right) \subset \ell^2 \left( \Gamma \right)$ is encoded in $\cB_n \left( \uppsi \right)$, for large $n \in \N_+$. On the one hand, this is not surprising, as \cref{eq:chi_expected_value} implies that representations of large rank separate elements of $\Gamma$. More precisely, for any $\gamma \in \Gamma$, there exists $n_\gamma \in \N_+$, such that for all $n \in \N_+ \cap \left[ n_\gamma, \infty \right)$, there exists $\varrho \in \Hom_{\irr} \left( \Gamma, \rU (n) \right)$, such that $\gamma \notin \ker \left( \varrho \right)$. On the other hand, this is contrast to the Euclidean case, where only the 1-dimensional representations are important. Quantifying the amount of information encoded in $\cB_n \left( \uppsi \right)$ and extending the above results to all of $\ell^2 \left( \Gamma \right)$ are the most important and interesting directions for further research.

\bigskip

\section{The hyperbolic Bloch transform}
\label{sec:HBT}

In this section we study the case when the functions, $\uppsi$, are complex functions on the hyperbolic plane. In this case $\cH_0$ in \cref{eq:abt_with_modes} can be chosen to be the space of (square integrable) functions on a fundamental cell for $\Gamma$ in $\bH$. However, in order to better understand the Bloch transform, we modify \cref{eq:abt_with_modes}, to get a geometrically more meaningful transform. This is achieved by the fact that for closed surfaces the Riemann--Hilbert map is an isomorphism between irreducible, finite dimensional unitary representations and flat, finite rank, irreducible, Hermitian vector bundles; cf. \cites{donaldson_anti_1985}.

\subsection{The wave functions on the hyperbolic plane}
\label{sec:wave}

So far we regarded $\Gamma$ as a group abstractly defined by \cref{eq:relation}. For the rest of this paper, we fix a concrete realization of $\Gamma$ as the fundamental group of a closed surface, $\Sigma$, of genus $g \in \N \cap \left[ 2, \infty \right)$ with a fixed hyperbolic Riemannian metric.

Note that the above setup immediately yields a faithful embedding of $\Gamma$ into the isometry group of the hyperbolic plane, making it a Fuchsian group. Hence, $\Gamma$ has canonical unitary and faithfully representations on the Sobolev spaces, $L_k^p \left( \bH \right)$, for any $k \in \N$ and $p \in \left[ 1, \infty \right) \cup \{ \infty \}$, via $\gamma \mapsto T_\gamma \eqdef \left( f \mapsto f \circ \gamma^{- 1} \right)$. Now we have the following version of the Bloch transform in \eqref{eq:abt_with_modes}: Let $\fC \subset \bH$ be a fundamental cell for the action of $\Gamma$ and for all $\left( \uppsi, \varrho, x \right) \in C_\cpt^\infty \left( \bH \right) \times \Hom_{\irr} \left( \Gamma, \rU (n) \right) \times \fC$, let
\begin{equation}
   \widetilde{\cB}_n \left( \uppsi, \varrho \right) (x) \eqdef \sum\limits_{\gamma \in \Gamma} \uppsi \left( \gamma (x) \right) \varrho \left( \gamma \right) \in L^2 \left( C \right) \otimes \End \left( \cx^n \right),
\end{equation}
and it is easy to see that if $\varrho \in \rhoeq$, then
\begin{equation}
   \left[ \varrho, \widetilde{\cB}_n \left( \uppsi, \varrho \right) \right] \in L^2 \left( C \right) \otimes \cV_{\rhoeq}^n,
\end{equation}
is independent of the choice of representative. As in the Euclidean case, $\widetilde{\cB} \left( \uppsi, \varrho \right)$ can be viewed as a \emph{quasi-periodic function}, except now the $\widetilde{\cB} \left( \varrho, \uppsi \right)$ is matrix-valued and the quasi-periodicity takes the form of
\begin{equation}
   T_\gamma \left( \widetilde{\cB} \left( \uppsi, \varrho \right) \right) = \widetilde{\cB} \left( \varrho, \uppsi \right) \varrho \left( \gamma \right).
\end{equation}
In particular, $\left| \widetilde{\cB} \left( \uppsi, \varrho \right) \right|$ is both $\Gamma$-periodic and independent of the choice of $\varrho$ within $\rhoeq$. Now \Cref{theorem:asymptotic_almost_unitarity} implies that for any smooth, compactly-supported function, $\uppsi$, we have that
\begin{equation}
   \lim\limits_{n \rightarrow \infty} \int\limits_{\cM_\Gamma^n} \left( \int\limits_C \left| \widetilde{\cB}_n \left( \uppsi, \varrho \right) \right|^2 \vol_\bH \right) \rd \mu_n = \| \uppsi \|_{L^2 \left( \bH \right)}^2.
\end{equation}
In the next section we define another (equivalent) version of $\widetilde{\cB}$ whose values are sections of certain endomorphism bundles over $\Sigma$.

\medskip

\subsection{Stable bundles and the hyperbolic Bloch transform}
\label{sec:stable}

Fix $x_0 \in \bH$ and $n \in \N_+$. Let $\pi : \bH \rightarrow \Sigma$ be the projection (the factor map), and $y_0 \eqdef \pi \left( x_0 \right) \in \Sigma$. Let furthermore $\nabla^0$ be the product metric on $E^{(n)} \eqdef \Sigma \times \cx^n$ (with its standard Hermitian structure).

Let $\varrho \in \Hom_{\irr} \left( \Gamma, \rU (n) \right)$. By \cite{donaldson_anti_1985}, there exists a flat, irreducible, and metric compatible connection, $\nabla$, whose holonomy representation at $y_0 \eqdef \pi \left( x_0 \right)$ in the standard orthonormal frame of $E^{(n)}$ is $\varrho$. We write $\varrho_\nabla$ from now on to emphasize the connection. Furthermore, this connection is unique in Coulomb gauge with respect to $\nabla^0$. Let $A_\nabla \eqdef \nabla - \nabla^0$ be the connection 1-form and $\fU_\nabla : \bH \rightarrow \rU (n)$ be the (unique) solution to the following initial value problem:
\begin{align}
   \rd \fU_\nabla    &= - \pi^* \left( A_\nabla \right) \fU_\nabla, \\
   \fU_\nabla (x_0)  &= \id_n.
\end{align}
Now $\fU_\nabla$ can be interpreted as the pullback of the change-of-basis matrix for the parallel transport map of $\nabla$. It is easy to see now that $\fU_\nabla$ satisfies
\begin{equation}
   \fU_\nabla \circ \gamma = \varrho_\nabla \left( \gamma \right) \fU_\nabla\label{eq:equivariance}
\end{equation}
for all $\gamma \in \Gamma$. If $u \in \rU (n)$ is a smooth function and $\nabla^u \eqdef u \circ \nabla \circ u^*$, then
\begin{equation}
   \fU_{\nabla^u} = u \fU_\nabla u (x_0)^*. \label{eq:gauge_change_for_U}
\end{equation}

\smallskip

Let $\cC_\Gamma^n$ be the space of irreducible, flat, and metric compatible connections on $E^{(n)}$ that are in Coulomb gauge with respect to $\nabla^0$. By \cite{donaldson_anti_1985}, $\cC_\Gamma^n$ is canonically isomorphic to $\Hom_{\irr} \left( \Gamma, \rU (n) \right)$ via $\nabla \mapsto \varrho_\nabla$. In fact the usual actions of $\rU (n)$ commute with this map, thus it factors down to an isomorphism of the moduli space of irreducible, flat, and metric compatible connections on $E^{(n)}$ and $\cM_\Gamma^n$. For each $\nabla$ in $\cC_\Gamma^n$, let the corresponding equivalence class be $\nablaeq$. In particular, under the above identification of $\cM_\Gamma^n$ and $\cC_\Gamma^n / \rU (n)$, we have $\nablaeq \cong \widehat{\varrho}_\nabla$.

\smallskip

Next we define another bundle of Hilbert spaces, $\cE^n$, over $\cM_\Gamma^n$ that serves the role of $\cV^n$ for our hyperbolic Bloch transform. Let $\nabla \in \cC_\Gamma^n$ and define
\begin{equation}
   \widetilde{\cE}_\nabla^n \eqdef L^2 \left( \Sigma, \End \left( E^{(n)} \right) \right),
\end{equation}
which defines a Hilbert bundle using the pointwise Hilbert--Schmidt norm and the Riemannian metric on $\Sigma$. Now for all $u \in \rU (n)$, $\uppsi \in \widetilde{\cE}_\nabla^n$, let
\begin{equation}
   u \left( \varrho, \uppsi \right) \eqdef \left( u \varrho u^*, u \uppsi u^* \right) \in \widetilde{\cE}_{\nabla^u}^n.
\end{equation}
Let $\cE^n \eqdef \widetilde{\cE}^n / \rU (n)$ which is a Hilbert bundle over $\cM_\Gamma^n$ and let $L^2 \left( \cE^n \right)$ be the Hilbert space of square integrable sections of $\cE^n$ (with respect to $\mu_n$). Finally, let $\cE$ be the Hilbert-sheaf of $\cM_\Gamma$ induced by the vector bundles, $\cE^n$, and $L^2 \left( \cE \right) = \oplus_{n \in \N_+} L^2 \left( \cE^n \right)$.

For each $x \in \bH$, let the covector, $\tau_x \in \cE^*$, be defined via
\begin{equation}
   \tau_x \left( \left[ \nabla, \Phi \right] \right) \eqdef \tr \left( \fU_\nabla (x)^* \Phi \left( \pi (x) \right) \right).
\end{equation}

\smallskip

Finally, note that any vector $v$ on $\Sigma$ can be pulled back to a vector field $\widetilde{v}$ on $\bH$ using $\pi$.

\smallskip

We combine the above constructions to define the hyperbolic Bloch transform.

\begin{definition}[The hyperbolic Bloch transform]
   Let us define the map $\fB : C_\cpt^\infty \left( \bH \right) \rightarrow L^2 \left( \cE \right)$ such that for all $\uppsi \in C_\cpt^\infty \left( \bH \right)$, $\nablaeq \in \cM_\Gamma$, and $y \in \Sigma$, the section $\fB \left( \uppsi \right)$ satisfies
   \begin{equation}
      \fB \left( \uppsi \right) \left( \nablaeq \right) \left( y \right) \eqdef \left[ \nabla, \sum\limits_{x \in \pi^{- 1} \left( y \right)} \uppsi (x) \fU_\nabla (x) \right]. \label{eq:HBT_def}
   \end{equation}
\end{definition}

As before, the right-hand side of \cref{eq:HBT_def} is meaningful for smooth, compactly-supported function, $\uppsi$, and the ideas of the proof \Cref{theorem:asymptotic_almost_unitarity} again yield that, for each $n \in \N_+$, the map $\fB_n$, defined via $\uppsi \mapsto \fB_n \left( \uppsi \right) \eqdef \fB \left( \uppsi \right)|_{\cM_\Gamma^n} (\cdot)$ is a continuous linear map between the Banach spaces
\begin{equation}
   \left\{ \: \uppsi \in L^2 \left( \bH \right) \: \middle| \: \| \uppsi \| \eqdef \sum\limits_{\gamma \in \Gamma} \| \uppsi \|_{L^2 \left( \gamma (C) \right)} < \infty \: \right\} \cong \ell^1 \left( \Gamma \right) \otimes L^2 \left( C \right), \label{eq:Banach}
\end{equation}
and $L^2 \left( \cE \right)$. We prove this claim, and more, in \Cref{theorem:HBT} below.

When regarded as a densely defined operator from $L^2 \left( \bH \right)$ to $L^2 \left( \cE^n \right)$, the adjoint of $\fB_n$ is given by
\begin{equation}
   \fB_n^* \left( \widehat{\uppsi} \right) = \int\limits_{\cM_\Gamma^n} \tau \left( \widehat{\uppsi} \circ \pi \right) \rd \mu_n.
\end{equation}
More verbosely, if given $y \in \Sigma$ and choices of representatives $\nabla \in \nablaeq$, and we write $\widehat{\uppsi} \left( \nablaeq \right) \left( y \right) = \left[ \nabla, \Phi_\nabla \left( y \right) \right]$, then for all $x \in \bH$
\begin{equation}
   \fB_n^* \left( \widehat{\uppsi} \right) \left( x \right) = \int\limits_{\cM_\Gamma^n} \tr \left( \fU_\nabla (x)^* \Phi_\nabla \left( \pi (x) \right) \right) \rd \mu_n \left( \nablaeq \right).
\end{equation}

\smallskip

We call a vector field, $w$, on $\bH$, \emph{$\Gamma$-periodic}, if for all $x \in \bH$ and $\gamma \in \Gamma$, $\pi_* \left( w (x) \right) = \pi_* \left( w \left( \gamma (x) \right) \right) \in T_{\pi (x)} \Sigma$. If $w$ is $\Gamma$-periodic, then there is a well-defined and unique vector field, $v$, on $\Sigma$, such that if $\pi (x) = y$, then $\pi_* \left( w (x) \right) = v \left( y \right)$. Since $\pi$ is a covering map, this correspondence is an isomorphism between the vector space of $\Gamma$-periodic vector fields on $\bH$ and the vector space of vector fields on $\Sigma$.

Finally, let us define three operators. First, if $\left[ \nabla, \Phi \right] \in \cE_{\nablaeq}$ and $v$ is a vector field on $\Sigma$, then let
\begin{equation}
   \whnabla_v \left( \left[ \nabla, \Phi \right] \right) \eqdef \left[ \nabla, \nabla_v \Phi \right].
\end{equation}
Then $\whnabla_v$ induces a bundle map of $\cE$.

Second, let
\begin{equation}
   \whDelta \left( \left[ \nabla, \Phi \right] \right) \eqdef \left[ \nabla, \nabla^* \nabla \Phi \right].
\end{equation}
Then $\whDelta$ induces a bundle map of $\cE$ as well.

Third, for all $\nabla \in \cC_\Gamma^n$ and $\gamma \in \Gamma$, let $\Hol_{\nabla, \gamma}$ be the bundle map of $E^{(n)}$, defined as follows:
\begin{equation}
   \forall \varphi \in E^{(n)} : \quad \Hol_{\nabla, \gamma} \left( \varphi \right) \eqdef \varrho_\nabla \left( \gamma \right) \varphi
\end{equation}
Then $\Hol_{\cdot, \gamma}$ induces a bundle map, $\whHol_\gamma$, of $\cE$ as well.

\smallskip

Now we are ready to state and prove our second main result:

\begin{theorem}
   \label{theorem:HBT}
   For each $\uppsi \in L^2 \left( \bH \right)$ with $\| \uppsi \| < \infty$ (defined in \cref{eq:Banach}), the sequence, $\left( \fB_n^* \fB_n \left( \uppsi \right) \right)_{n \in \N_+}$, converges to $\uppsi$ in the topology of $L^\infty \left( \bH \right)$.

   Furthermore, if $v$ is a vector field on $\Sigma$ with induced $\Gamma$-periodic vector field $w$ on $\bH$, then
   \begin{equation}
      \fB \circ \rd_w = \whnabla_v \circ \fB, \label{eq:bloch_derivative}
   \end{equation}
   holds on $C_\cpt^\infty \left( \bH \right)$.

   Finally, we have that
   \begin{subequations}
   \begin{align}
      \fB \circ \Delta     &= \whDelta \circ \fB, \label{eq:Delta_commutator} \\
      \fB \circ T_\gamma   &= \whHol_\gamma \circ \fB, \label{eq:Hol_commutator}
   \end{align}
   \end{subequations}
   holds on $C_\cpt^\infty \left( \bH \right)$.
\end{theorem}

\begin{proof}
   The proof of \Cref{theorem:asymptotic_almost_unitarity} can be adapted to the first claim.

   By density, it is enough to show the second statement for smooth vector fields. Let $v$ be a smooth vector field on $\Sigma$ with induced $\Gamma$-periodic (smooth) vector field $w$ on $\bH$ and let $\uppsi$ be a smooth, compactly-supported function on $\bH$. Then for all $y \in \Sigma$, we have
   \begin{align}
      \left( \fB \circ \rd_w \right) \left( \uppsi \right) \left( \nablaeq \right) \left( y \right)  &= \fB \left( \rd_w \uppsi \right) \left( \nablaeq \right) \left( y \right) \\
         &= \left[ \nabla, \sum\limits_{x \in \pi^{- 1} \left( y \right)} \left( \rd_w \uppsi \right) (x) \fU_\nabla (x) \right] \\
         &= \left[ \nabla, \sum\limits_{x \in \pi^{- 1} \left( y \right)} \left( \rd_w \left( \uppsi \fU_\nabla \right) + \pi^* \left( A_\nabla \right) (w) \left( \uppsi \fU_\nabla \right) \right)(x) \right] \\
         &= \left[ \nabla, \sum\limits_{x \in \pi^{- 1} \left( y \right)} \left( \rd_w + A_\nabla \left( \pi_* (w) \right) \right) \left( \uppsi \fU_\nabla \right) (x) \right] \\
         &= \left[ \nabla, \left( \nabla_v^0 + A_\nabla \left( v \right) \right) \sum\limits_{x \in \pi^{- 1} \left( y \right)} \left( \uppsi \fU_\nabla \right) (x) \right] \\
         &= \left[ \nabla, \nabla_v \sum\limits_{x \in \pi^{- 1} \left( y \right)} \left( \uppsi \fU_\nabla \right) (x) \right] \\
         &= \left( \whnabla_v \circ \fB \right) \left( \uppsi \right),
   \end{align}
   which proves \cref{eq:bloch_derivative}.

   Now let $\uppsi$ still be a smooth, compactly-supported function on $\bH$. Fix $x_0 \in \cH$ and let $C$ be the Wigner--Seitz cell of $\Gamma$, centered at $x_0$, that is, if $\dist_\bH : \bH \times \bH \rightarrow \left[ 0 \infty \right)$ is the hyperbolic distance function, then
   \begin{equation}
      C_{x_0} \eqdef \left\{ \: x \in \bH \: \middle| \: \forall \gamma \in \Gamma - \{ e \} : \: \dist_\bH \left( x, x_0 \right) < \dist_\bH \left( x, \gamma \left( x_0 \right) \right) \: \right\}.
   \end{equation}
   By choosing a partition of unity for $\bH$, we can assume, without any loss of generality, that $\supp \left( \uppsi \right) \subset C_{x_0}$. Let $\varphi : \bH \rightarrow [0, 1]$ be a smooth cutoff function, just that $\supp \left( \uppsi \right) \subseteq \varphi^{- 1} (1)$ and $\supp \left( \varphi \right) \subset C_{x_0}$. Let $\left( x_1, x_2 \right) : \bH \rightarrow \mathbbm{D}$ be the global chart for $\bH$ coming from the Poincar\'e disk model, centered at $x_0$, and let $\partial_1$ and $\partial_2$ be the corresponding vector fields. Then $\Delta = - \left( 1 - x_1^2 - x_2^2 \right)^2 \left( \partial_1^2 + \partial_2^2 \right)$. Next we define a $\Gamma$-periodic (smooth) vector fields $w_1$ and $w_2$ as follows: If $x \in \gamma \left( C_{x_0} \right)$, then $w_i (x) = \varphi \left( \gamma^{- 1} (x) \right) \partial_i (x)$, and zero otherwise. Let $v$ be the induced Note that $\pi \left( C_{x_0} \right)$ is open and dense in $\Sigma$, so let $r$ be a smooth function on $\Sigma$, such that $r \circ \pi = - \varphi \left( 1 - x_1^2 - x_2^2 \right)^2$ holds on $C_{x_0}$ (such an $r$ exists and unique). Then we have that
   \begin{equation}
      \Delta \uppsi = \left( r \circ \pi \right) \left( \rd_{w_1}^2 + \rd_{w_2}^2 \right) \uppsi,
   \end{equation}
   and if $\left[ \nabla, \Phi \right] \in \cE_{\nablaeq}^n$ so that $\supp \left( \Phi \right) \subset \pi \left( \varphi^{- 1} (1) \right)$, then
   \begin{equation}
      \whDelta \left[ \nabla, \Phi \right] = r \left( \whnabla_{v_1}^2 + \whnabla_{v_2}^2 \right) \left[ \nabla, \Phi \right] = \left[ \nabla, r \left( \nabla_{v_1}^2 + \nabla_{v_2}^2 \right) \Phi \right].
   \end{equation}
   Thus for all $y \in \Sigma$ and $\nablaeq \in \cM_\Gamma$, using \cref{eq:bloch_derivative}, we have
   \begin{align}
      \left( \fB \circ \Delta \right) \left( \uppsi \right) \left( \nablaeq \right) \left( y \right) &= \left[ \nabla, \sum\limits_{x \in \pi^{- 1} \left( y \right)} r \left( y \right) \left( \left( \rd_{w_1}^2 + \rd_{w_2}^2 \right) \uppsi \right) (x) \fU_\nabla (x) \right] \\
         &= \left[ \nabla, r \left( y \right) \left( \nabla_{v_1}^2 + \nabla_{v_2}^2 \right) \sum\limits_{x \in \pi^{- 1} \left( y \right)} \uppsi (x) \fU_\nabla (x) \right] \\
         &= r \left( \whnabla_{v_1}^2 + \whnabla_{v_2}^2 \right) \fB \left( \uppsi \right) \left( \nablaeq \right) \left( y \right) \\
         &= \left( \whDelta \fB \left( \uppsi \right) \right) \left( \nablaeq \right) \left( y \right),
   \end{align}
   which proves \cref{eq:Delta_commutator}.

   The proof of \cref{eq:Hol_commutator} is straightforward from the definitions.
\end{proof}

\medskip

\subsection{Periodic magnetic fields and twisted bundles}
\label{sec:twisted}

Let us assume that a periodic magnetic field is present on $\bH$. More precisely, let $a$ be an imaginary valued 1-form, such that for all $\gamma \in \Gamma$ we have that
\begin{equation}
   T_\Gamma^* (a) = a + u_\gamma \rd \left( u_\gamma \right),
\end{equation}
for some $u_\gamma : \Gamma \rightarrow \rU (1)$. In other words, the connection $D \eqdef \rd + a$ is gauge equivalent to the pullback of a unitary connection, $\nabla^\cL$, on a Hermitian line bundle, $\cL \rightarrow \Sigma$. In particular, if $F_D$ and $F_{\nabla^\cL}$ are the curvature 2-forms of $D$ and $\nabla^\cL$, respectively, then by Chern--Weil theory, for any fundamental cell, $C \subset \bH$ of $\Gamma$, we have
\begin{equation}
   \textsc{flux} \eqdef \tfrac{1}{2 \pi i} \int\limits_C F_D = \tfrac{1}{2 \pi i} \int\limits_C \rd a = \tfrac{1}{2 \pi i} \int\limits_\Sigma F_{\nabla^\cL} \in \Z.
\end{equation}
Let $\varphi$ be an isomorphism of the Hermitian bundles $\underline{\cx}$ and $\pi^* \left( \cL \right)$, such that $D = \varphi^{- 1} \circ \nabla^\cL \circ \varphi$.

With this in mind, for all $\nabla \in \cC_\Gamma^n$, we let
\begin{equation}
   \widetilde{\cE}_\nabla^{n, \cL} \eqdef L^2 \left( \Sigma, \End \left( E^{(n)} \right) \otimes \cL \right),
\end{equation}
and $\cE^{n, \cL} \eqdef \widetilde{\cE}^{n, \cL} / \rU (n)$, which is again a Hilbert bundle over $\cM_\Gamma^n$. Now let us define the \emph{twisted hyperbolic Bloch transform} via
\begin{equation}
   \fB^\cL \left( \uppsi \right) \left( \nablaeq \right) \left( y \right) \eqdef \left[ \nabla, \sum\limits_{x \in \pi^{- 1} \left( y \right)} \varphi (x) \left( \uppsi (x) \right) \fU_\nabla (x) \right].
\end{equation}
It is easy to verify that, when restricted to $\cM_\Gamma^n$, $\fB^\cL \left( \uppsi \right)$ takes values in $\cE^{n, \cL}$ and the results of \Cref{theorem:HBT} can be proven with $\rd$ replaced by $D$.

\begin{remark}
   Similarly, given any finite rank Hermitian vector bundle, $(V, h)$, over $\Sigma$, and unitary connection, $\nabla^V$, on $V$, one can twist the Bloch transform by $\left( V, h, \nabla^V \right)$, and prove a new version of \Cref{theorem:HBT}. We also remark that $\widetilde{\cE}_\nabla^{n, \cL}$ is a (square integrable) \emph{Higgs field} for the bundle $E^{(n)}$. This appearance of Higgs bundles in the noncommutative Bloch transform adds another direction for the application of Higgs bundles to hyperbolic band theory as in \cite{KR22}.
\end{remark}

   \bibliography{references}

\end{document}